\documentclass[10pt， conference， letterpaper]{IEEEtran}
\IEEEoverridecommandlockouts
\usepackage{booktabs}
\usepackage{siunitx}
\usepackage{graphicx}
\usepackage{subcaption}
\usepackage{cite}
\usepackage{amsthm,amsmath,amsfonts,amssymb}
\usepackage{algorithmic}
\usepackage{algorithm}
\newtheorem{theorem}{Theorem}
\usepackage{algorithmic}
\usepackage{graphicx}
\usepackage{textcomp}
\usepackage{xcolor}
\usepackage[justification=centering]{caption}
\def\BibTeX{{\rm B\kern-.05em{\sc i\kern-.025em b}\kern-.08em
    T\kern-.1667em\lower.7ex\hbox{E}\kern-.125emX}}

\begin{document}

\title{CountingStars: Low-overhead Network-wide Measurement in LEO Mega-constellation Networks}

\author{
	\IEEEauthorblockN{
		Xiyuan Liu\IEEEauthorrefmark{1}, 
		Guanzuo Liu\IEEEauthorrefmark{1}, 
		Xiucheng Tian\IEEEauthorrefmark{1}, 
		and Wenting Wei\IEEEauthorrefmark{1} \\}
	\IEEEauthorblockA{\IEEEauthorrefmark{1}School of Telecommunication Engineering, Xidian University, Xi'an, China\\
    markliu225@stu.xidian.edu.cn, wtwei@xidian.edu.cn}
}

\maketitle
\thispagestyle{empty}
\pagestyle{empty}
\begin{abstract}

The high mobility of satellites in Low Earth Orbit (LEO) mega-constellations induces a highly dynamic network topology, leading to many problems like frequent service disruptions. To mitigate this, Packet-based Load Balancing (PBLB) is employed. However, this paradigm shift introduces two critical challenges for network measurement stemming from the requirement for port-level granularity: memory inflation and severe hash collisions. To tackle these challenges, we propose CountingStars, a low-overhead network-wide measurement architecture. 
In the ground controller, CountingStars builds a digital twins system to accurately predict the future network topology. This allows ground controller to generate and distribute collision-free hash seeds to satellites in advance. On the satellite, we introduce a port aggregation data structure that decouples the unique flow identifier from its multi-port counter and updates it through efficient bit operations, solving the memory inflation caused by PBLB.
Simulation results show that the memory usage of CountingStars is reduced by 70\% on average, and the relative error of measurement is reduced by 90\% on average.
Implementation on FPGA shows its prospect to deploy in real system. 
\end{abstract}


\section{Introduction}
The rapid deployment of LEO mega-constellations has revolutionized global connectivity by enabling ultra-low-latency, high-throughput communication \cite{intro1}. However, the inherent high-dynamic topology and unstable interconnection of these networks pose unprecedented challenges for traffic management \cite{intro2, intro3}. Traditional load balancing strategies, designed for static or slowly evolving networks, fail to adapt to such dynamic environments, leading to cascading congestion, increased packet loss, degraded quality of service (QoS) and even frequent service disruptions.


To address these problems, packet based load balancing (PBLB) has emerged as a critical paradigm for LEO networks \cite{pblb1, pblb2, pblb3}. Unlike coarse-grained approaches, packet level methods dynamically distribute packets of a single flow across multiple ports in a fine-grained manner, enhancing both link utilization and network robustness \cite{intro4, intro5, intro6, intro7, intro8, intro9}. However, achieving such granularity requires port-level network measurement to monitor real-time link states and traffic patterns. Such demands pose one serious challenge for network-wide measurement in scenarios with limited memory resources.


Another challenge to network-wide measurement is that, frequent topology changes force flow path that should be stable to dynamically change the passing satellite nodes over time. From the perspective of any single satellite, the set of traffic it carries is constantly and rapidly changing, forming a non-stationary traffic pattern\cite{intro10, intro11, intro12}.
Measurement schemes that relies on static hashing strategies will fail in this environment. The hash function that performs well in one cycle may experience severe hash collisions in the next cycle due to new and unrelated flow being unpredictably mapped to the same memory address, thereby reducing measurement accuracy.




Existing approaches such as Satformer \cite{temp1} demonstrate adaptability to dynamic network scenarios but suffer from poor generalization capabilities. This limitation introduces critical failure risks, particularly in real-world systems where measurement inaccuracies can propagate catastrophic errors (e.g., routing misconfigurations or resource overallocation). Furthermore, In-band Network Telemetry (INT) \cite{int1,int2} imposes excessive computational overhead on satellite nodes. Per-packet processing requires sequential operations—header parsing, telemetry metadata insertion, and checksum recalculation—which strain resource-constrained satellite hardware, rendering INT impractical for large-scale LEO deployments.

While efficient measurement techniques like Count-Min Sketch (CM) \cite{intro13}, Elastic Sketch (ES) \cite{intro14}, and FlowLIDAR \cite{intro15} achieve remarkable performance in terrestrial networks through flow aggregation (e.g., heavy-hitter separation, dynamic compression), their direct adoption in LEO constellations faces two fundamental challenges:




\begin{enumerate}
    \item \textbf{Hash Collision Vulnerability}: Traditional hash functions struggle to adapt to LEO’s time-varying topology, causing frequent collisions that degrade measurement accuracy.
    \item \textbf{Port-Level Measurement Granularity}: Fine-grained port-level monitoring demands excessive on-board memory, conflicting with satellites’ stringent hardware constraints. Multi-instance deployment of sketch structures (e.g., replicating CM per port) causes memory consumption to scale linearly with port count, creating \textit{\textbf{a memory inflation bottleneck}}.
\end{enumerate}


To address these issues, we propose CountingStars, a network-wide measurement architecture tailored for dynamic LEO mega-constellations. Specifically, it presents a topology-aware dynamic hashing mechanism by digital twins to avoid hash collisions. To fine-grained measurement, compact port-aggregation data structure is designed to eliminate per-port counter instantiation and reducing memory overhead. Our contributions are listed below:


\begin{itemize}
    \item 
  By simulating satellite orbits and ISL states in real-time with the digital twins, CountingStars predicts future network states.  This allows pre-configure memory management policies before topology changes.

    \item 
   A topology-aware dynamic hash mechanism is proposed, so that ground controllers can predict satellite flow sets for upcoming epochs, compute minimal perfect hash seeds, and disseminate them to satellites. This eliminates hash collisions caused by traffic dynamics, improving the accuracy of the measurement. 
    \item 

    To reduce the memory usage and realize on-board deployment, we designed a port aggregation data structure: it divides a single counter into multiple sub-fields and can update the flow of a specific port through a jump-based bit operation. This eliminates the need to instantiate an independent data structure for each port. At the same time, to cope with high throughput, each port is equipped with multiple parsers, and the data packets are distributed in a round-robin mechanism.
\end{itemize}

The rest of this paper is organized as follows: Section II analyzes the limitations of existing solutions and motivates our design. Section III details the CountingStars architecture. Section IV evaluates its performance against baseline methods, and Section V concludes the paper.

\section{Motivation and Key Insights}
In LEO mega-constellations networks, PBLB is a key technology to improve network efficiency, but it also inevitably causes two core measurement problems: Memory Inflation and Hash Collision. These two problems seriously restrict the feasibility and accuracy of on-orbit measurement. In the following, the causes of the two core challenges are analyzed in detail and the solutions are discussed.
\subsection{Memory Inflation}
In LEO mega-constellation networks, per-packet load balancing is considered as a core mechanism. With per-packet equalization, a single flow is sent to multiple outgoing ports, resulting in N times the number of entries (N is the number of outgoing ports). This change causes the size of the traffic to be measured to surge to several times the original.

\begin{figure}[!htbp]
	\centering
	\includegraphics[width=\linewidth]{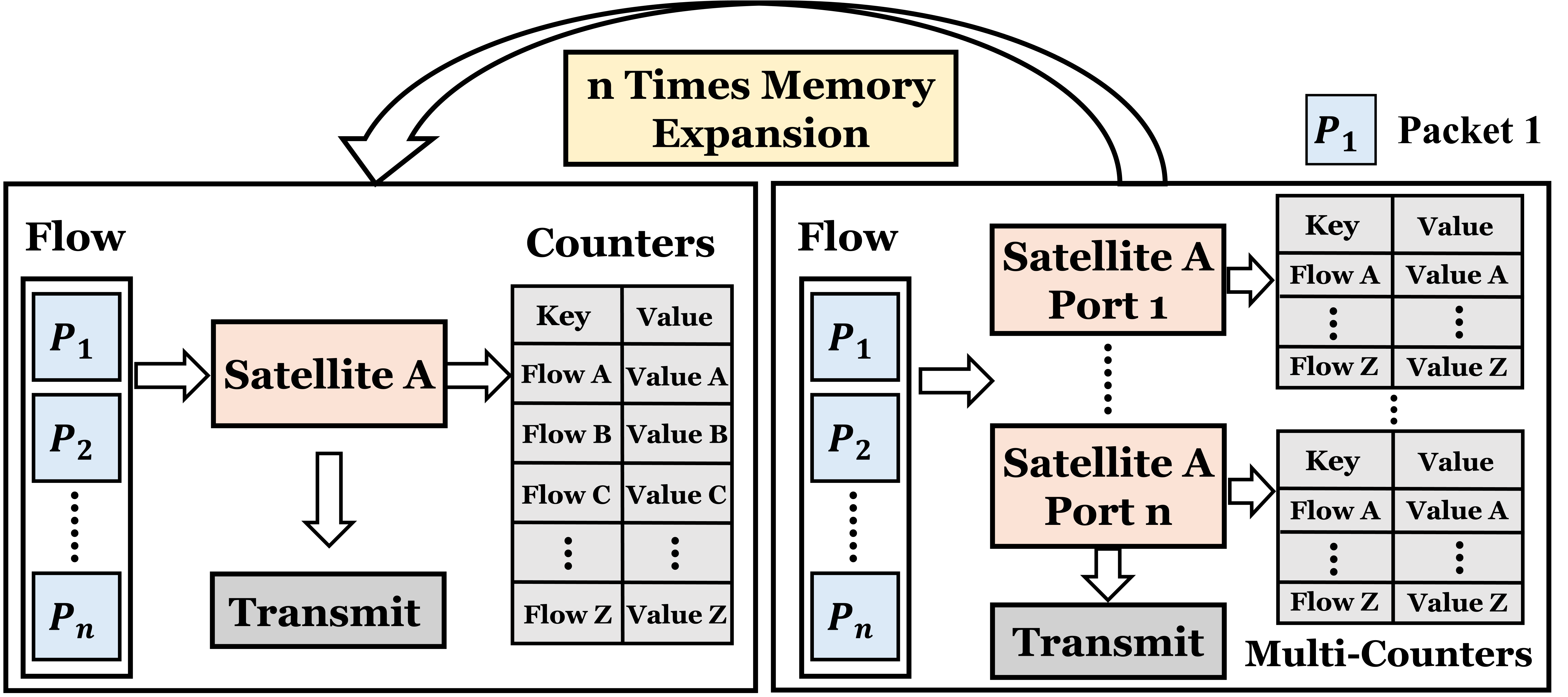}
	\caption{Cases of Memory Inflation}
	\label{Memory Inflation}
\end{figure}

If a flow goes through a fixed outgoing port only, then the measurement system only needs to maintain a record in memory for this flow\cite{intro16, intro17}, containing its flow ID and count result (As shown in Figure. \ref{Memory Inflation}).

In PBLB scenario, the packets of this flow may be distributed to $n$ different outgoing ports such as Port 1 to Port $n$. To achieve fine-grained measurements at the port level, the measurement system is forced to maintain independent measurement states at each outgoing port for this same flow. This means that the quintuple used to uniquely identify the flow must be duplicated $n$ times in memory, each associated with the respective counters from Port 1 to Port $n$.

\begin{figure}[!htbp]
	\centering
	\includegraphics[width=\linewidth]{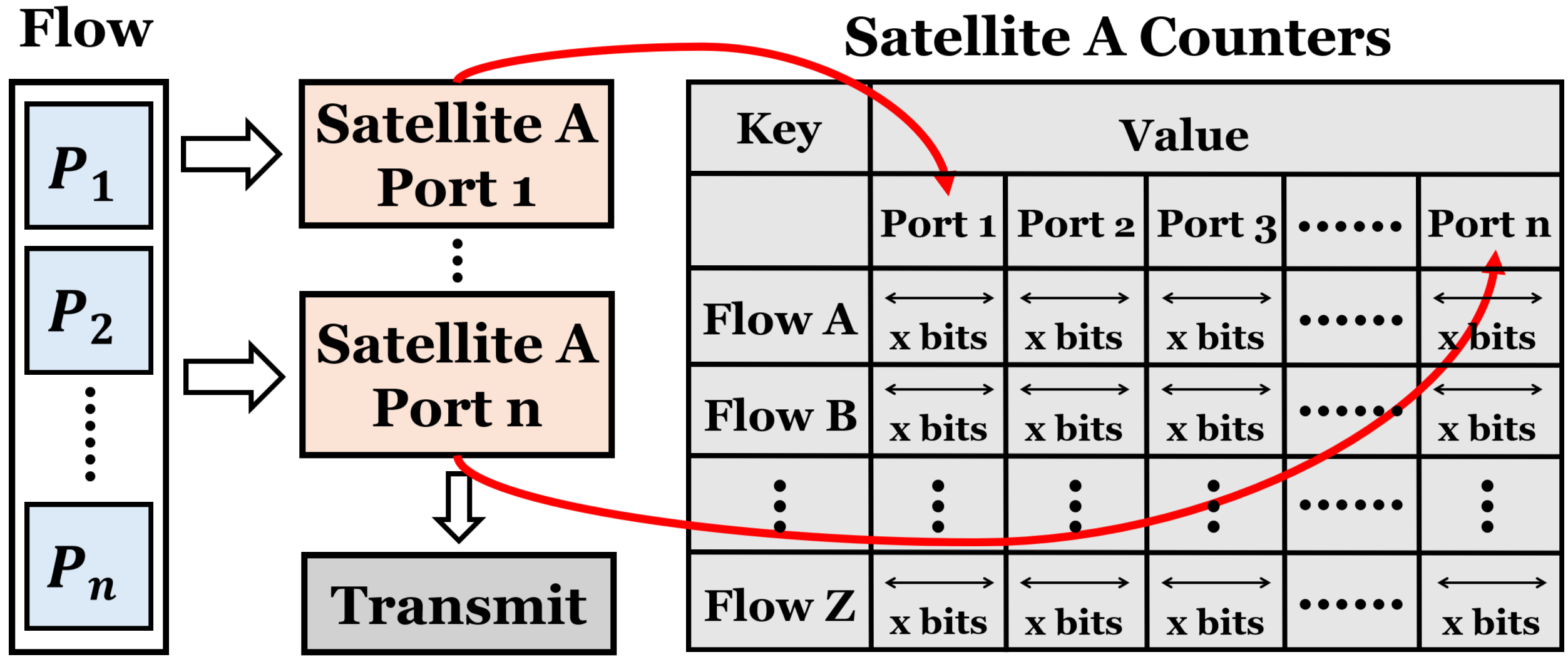}
	\caption{Solution of Memory Inflation}
	\label{Memory Inflation Solution}
\end{figure}
As shown in Figure. \ref{Memory Inflation Solution}, CountingStars stores only unique identifiers for each flow, associated with a compact collection of port counters. A large memory block (say $nx$ bits wide) is logically divided into multiple small counters ($nx$ bits). When a packet arrives, the system updates the $x$-bit counter at the corresponding port using jump-based counter update method.

\subsection{Hash Collision}

The highly dynamic topology of LEO constellation networks poses a challenge to static hash-based measurement. The rapid movement of satellites leads to frequent topology changes, so that flows at different time points may be mapped to the same hash address, which causes hash collisions and destroys the accuracy of traffic statistics.

\begin{figure}[!htbp]
	\centering
	\includegraphics[width=\linewidth]{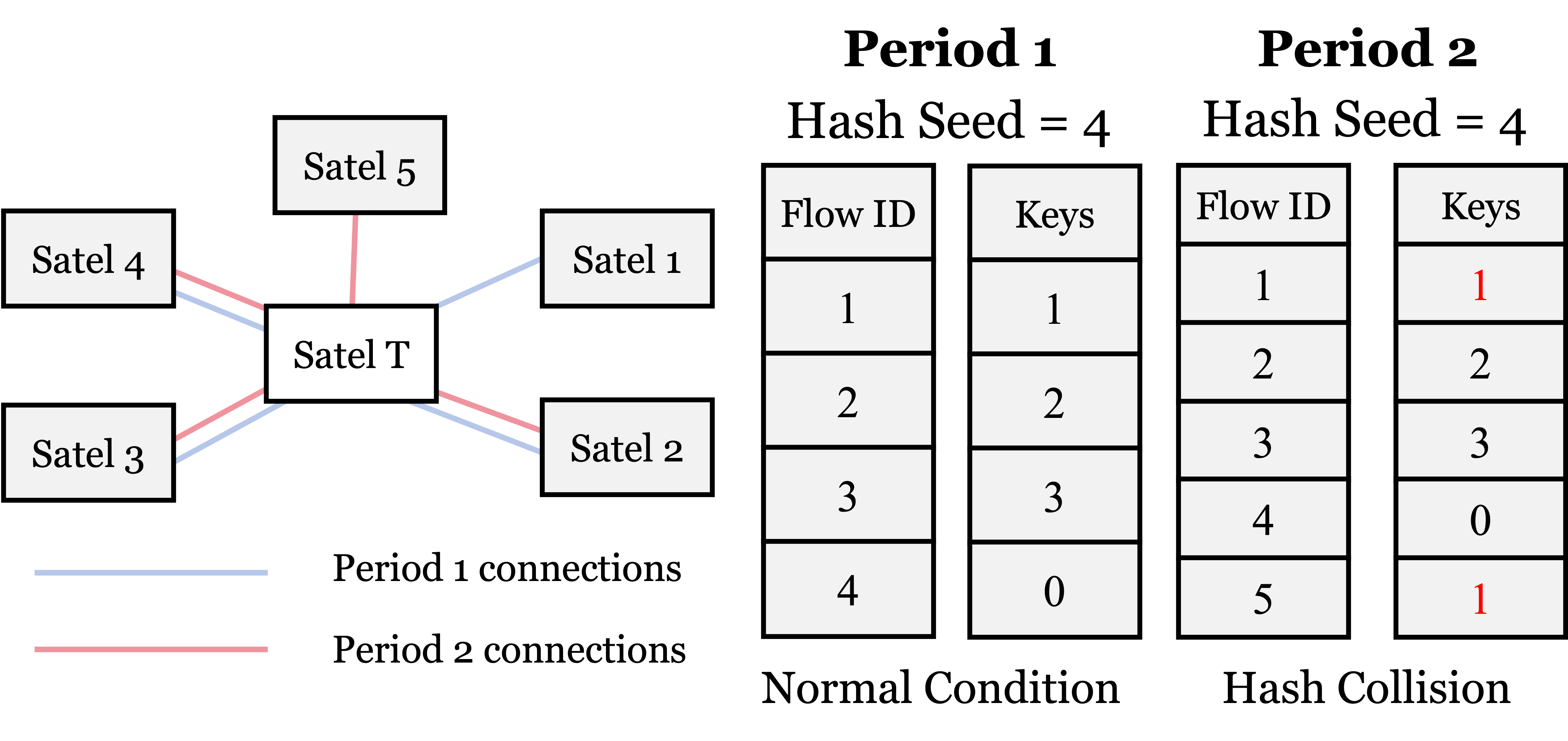}
	\caption{Cases of Hash Collision}
	\label{Hash Collision}
\end{figure}


Consider the scenario in Figure \ref{Hash Collision} at Period 1, the target satellite is connected to satellites 1, 2, 3, and 4, with a hash seed of 4, generating key values 0, 1, 2, and 3, each corresponding to a satellite’s flow. At Period 2, the target satellite’s connections expand to include satellites 2, 3, 4, and 5, with the hash seed remaining 4. Due to hash function limitations, flows from 1 and 5 map to the same key value 1. 

This hash collision across different flows at different times prevents the measurement system from distinguishing them, incorrectly aggregating traffic data from satellites 1 and 5. This issue is particularly severe in LEO networks, as constellation topologies undergo significant changes every few minutes, causing repeated hash address overlaps that continuously disrupt the accuracy of traffic records.

\begin{figure}[!htbp]
	\centering
	\includegraphics[width=\linewidth]{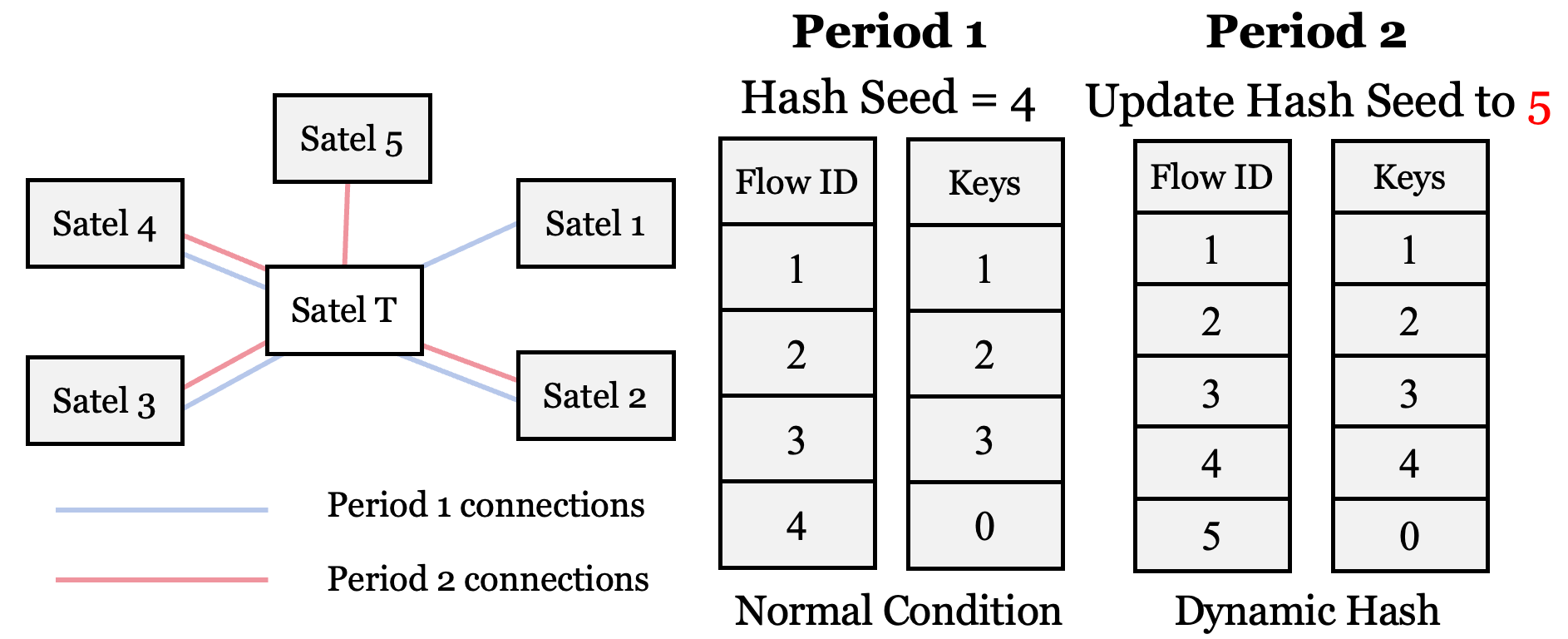}
	\caption{Solution of Hash Collision}
	\label{Dynamic Hash Solution}
\end{figure}

As shown in Figure. \ref{Dynamic Hash Solution}, we introduced dynamic hash mechanism. It can ensure that the occurrence of hash address overlap is avoided under frequent topological changes, while reducing memory and computational overhead, thereby achieving precise congestion control and QoS guarantee in mega-constellation networks.

\section{Design of CountingStars}

In LEO mega-constellation networks, the dynamic topology brought about by the high-speed movement of satellites and the limited on-board resources pose severe challenges to measurement \cite{limit1, limit2, limit3}. Existing methods cannot adapt to frequent topology changes, have low measurement accuracy, and cause memory inflation due to multiple instantiations. To this end, CountingStars is proposed as a low-overhead network-wide measurement architecture. Through the collaborative mechanism of ground controllers and satellites, it aims to address hash collisions and memory inflation, thereby achieving low-overhead fine-grained measurement.

\begin{figure}[!htbp]
	\centering
	\includegraphics[width=\linewidth]{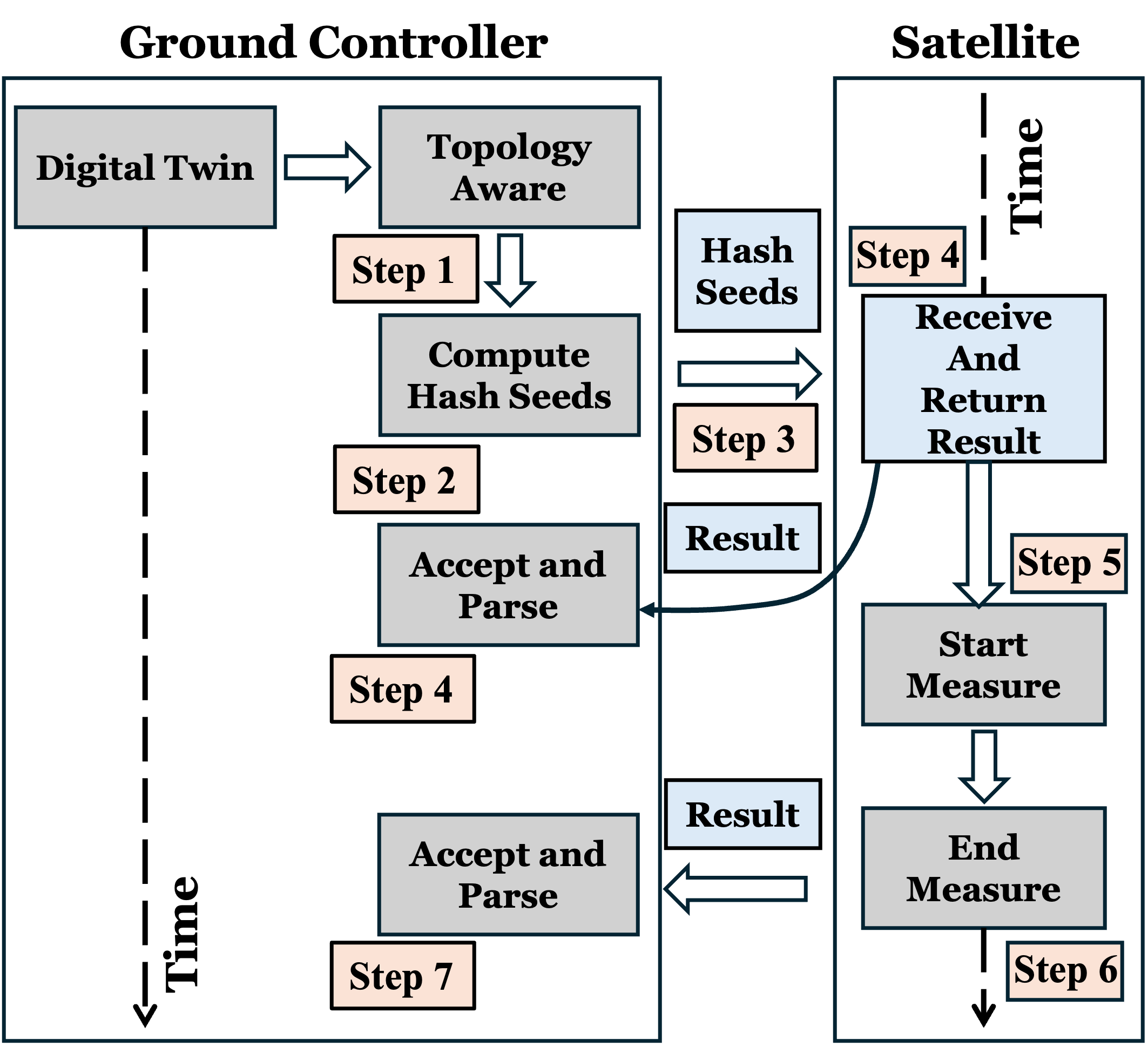}
	\caption{Workflow of CountingStars}
	\label{Workflow}
\end{figure}

Figure \ref{Workflow} shows the workflow of CountingStars. The ground controller contains the digital twin subsystem. This subsystem calculates the satellite position in real time through the orbital dynamics model, and uses the geometric visibility constraints to predict the future network topology. Thus, all possible flows in the next time period are obtained; then the unique identifier of each flow is calculated through the Cantor pairing function (Step1). Consequently, the collision-free hash seed can be calculated for the next measurement period (Step 2), and then the seed is distributed to the satellite (Step 3).

The satellite receives the hash seed and ensures that the satellite transmits the measurement results back before the next measurement starts (step 4). During the measurement period (steps 5-6), the satellite addresses and counts using hash seeds.  At the end of the period, the satellite transmits the measurement results back to the ground system, which receives and analyzes them (step 7) to complete a measurement closed loop. This architecture puts complex topology-aware and hash computations on the ground, and only performs table lookup and counting on the satellite, thus effectively solving the two core challenges.

To achieve the above design goals, the overall architecture of CountingStars consists of a series of key modules working together. The following will provide a detailed introduction to the functions and workflow of these modules.

\subsection{Digital Twins} \label{Digital Twins}


One keypoint of the CountingStars is a digital twins subsystem in the ground controller. It aims to construct a high-fidelity virtual model of the LEO mega-constellation, thereby providing precise prediction of the network's topological state for future time periods. To achieve this goal, we build two key components: the orbit modeling component accurately predicts the satellite position, and the network construction component generates the runtime topology by simulating the link construction strategy and physical constraints.


\subsubsection{Orbit Modeling Component}
The foundation of the digital twins is the modeling of satellite orbits, which begins with parsing ephemeris data. We utilize the simulation software Systems Tool Kit (STK) by importing publicly available Two-Line Element (TLE) sets. TLE provides six classic Keplerian elements that describe the satellite's orbit at a specific epoch. These elements collectively define the size, shape, and orientation of the orbit, as well as the position of the satellite at that specific moment. For orbit propagation, these elements are first converted into a Cartesian State Vector $[\vec{r_0}, \vec{v_0}]$, which represents the satellite's initial position and velocity in the Earth-Centered Inertial (ECI) frame \cite{orb1, orb2, orb3}.

With the initial state vector, we use the classic Two-Body model to propagate its orbit. This model makes an idealized assumption, treating the Earth as a perfectly spherical point mass and the sole source of gravity. Under this model, the satellite's motion follows a concise dynamical equation:

\begin{equation}\label{eq:Two-Body}
    \ddot{\vec{r}} + \frac{\mu}{r^3}\vec{r} = 0
\end{equation}
where $\ddot{\vec{r}}$ is the satellite's acceleration vector, $\vec{r}$ is its position vector, $r$ is its distance from the Earth's center, and $ \mu = GM $ is the Earth's standard gravitational parameter ($G$ being the gravitational constant and $M$ the Earth's mass). STK rapidly computes the satellite's future position vector $\vec{r}(t)$ and velocity vector $\vec{v}(t)$ by numerically integrating Equation \ref{eq:Two-Body}.

\subsubsection{Network Construction Component}
Based on the propagated position vector $\vec{r}(t)$, the digital twins system performs a network-wide visibility analysis. For any two satellites $i$ and $j$, a line-of-sight (LoS) is considered to exist only if their separation vector $\vec{r}_{ij}(t) = \vec{r}_j(t) - \vec{r}_i(t)$ is not occulted by the Earth. This geometric condition is met if the angle $\theta$ between the geocentric vector $(-\vec{r}_i(t))$ and the inter-satellite vector $(\vec{r}_{ij}(t))$ satisfies:

\begin{equation}
    \cos(\theta) = \frac{-\vec{r}_i(t) \cdot \vec{r}_{ij}(t)}{|\vec{r}_i(t)| |\vec{r}_{ij}(t)|} < \frac{\sqrt{|\vec{r}_i(t)|^2 - R_E^2}}{|\vec{r}_i(t)|}
\end{equation}
where $R_E$ is the radius of the Earth. By systematically performing this calculation at each discrete time step, a time-series of potential topology graphs $G_v(t)$ is generated, encompassing all geometrically possible links.

However, A potential visibility link does not equate to an actual communication link. Therefore, the system further integrates a satellite behavior simulation module to refine the potential topology into the operational topology $G_a(t)$. This module prunes links from the potential topology graph that do not meet realistic conditions by modeling a series of physical and operational constraints.

The module first models the physical hardware and basic linking strategy. It enforces hardware limitations, such as 4 laser transceivers on the Starlink V1.5 satellite\cite{starlink}, which directly restricts the maximum degree of a node in the network topology graph. Second, it applies Starlink's grid-like mesh linking strategy: each satellite maintains two intra-plane ISLs with its neighbors ahead and behind in the same orbital plane and establishes two inter-plane ISLs with corresponding satellites in adjacent planes.

The module also simulates more realistic operational rules. For instance, to avoid communication failure caused by extremely high relative speed, we introduces "seams" at the boundaries between ascending and descending orbital planes, where direct inter-plane links are not established. Furthermore, the model considers specific link management strategies in high-latitude regions to ensure continuous global coverage. By applying this comprehensive set of realistic constraints, the system can deduce the deterministic network topology for any future time $t$. Ultimately, this topological data is provided to the ground controller's other logical modules, serving as the critical input for the topology-aware dynamic hash mechanism.

\subsection{Flow Set Generation from Predicted Topology}

The goal of this process is, for any satellite node $S_k$ in the network and any future time period $t$, to accurately determine the set of all flows that may traverse that satellite, which we denote as $\mathcal{F}_{k,t}$. The primary input for this process is the future network topology predicted by the digital twins system.

Since the future traffic demand is unpredictable, we proactively take into account all possible flows in the network. This design makes our approach rely only on predictable network topology $G_t = (V, E_t)$ ingress.

Based on this, the ground controller executes a deterministic routing computation process. It iterates through every possible source-destination (SD) pair of nodes in the network, $(u, v) \in V \times V, u \neq v$, and calculates the shortest path between them on the predicted topology graph $G_{a,i}$ using a standard routing algorithm (e.g., Dijkstra). This path represents the sequence of satellites that the corresponding data flow will traverse during time period $t$.

Next, the controller performs flow set aggregation. It iterates through all computed paths. For a path $P_f$ corresponding to a flow $f=(src, dst)$, the controller identifies all satellite nodes contained in the path. For each node $S \in P_f$, the controller adds the flow $f$ to that satellite's flow set $\mathcal{F}_{k,t}$.

After iterating through all possible flows, the ground controller obtains a complete collection of per-satellite flow sets, $\{\mathcal{F}_{1,t}, \mathcal{F}_{2,t}, ..., \mathcal{F}_{|V|,t}\}$. This result serves as the direct and necessary input for the subsequent step of calculating the minimal collision-free hash seed for each satellite.

\subsection{Dynamic Hash Mechanism}
In LEO mega-constellation networks, the network topology is highly dynamic due to the rapid movement of satellite nodes. As mentioned above, this high dynamicity can lead to serious hash collisions and destroy the accuracy of traffic statistics. Therefore, we design a dynamic hash mechanism (DHM). The core idea is to abandon the static function, and instead, for each satellite in the network, in each discrete time period, dynamically generate an minimal collision-free hash function for a specific traffic set in that period.

This process begins on the ground controller: first, it leverages the future topology predicted by the digital twins system to generate the precise set of flows for each satellite; next, each flow is losslessly mapped to a unique integer identifier via the Cantor pairing function; then, for the set of integer identifiers corresponding to each satellite, a minimal perfect hash seed is computed to guarantee a collision-free mapping; finally, this minimal collision-free hash seed is distributed to the corresponding satellite to be activated for the upcoming measurement period.

Once the future network topology is predicted using the digital twin system, which in turn defines the exact set of flows $\mathcal{F}_{k,t}$ that each satellite will handle, the goal of DHM shifts to designing and applying an exclusive hash function $H_{k,t}$ for this known and concrete set. The function must satisfy two key properties:

\begin{itemize}
    \item Collision-free Hashing: two different flows in the set $\mathcal{F}_{k,t}$ can not be mapped to the same hash address;
    \item Minimal Hashing: the resulting hash address space must be as small as possible.
\end{itemize}
The construction of this dynamic hash function, $H_{k,t}$, is delicately decomposed into two separate mathematical steps that can be rigorously proved.

\subsubsection{Unique stream identifier generation}
We need to losslessly and reversibly convert the two-dimensional stream definition $f = (\text{src}, \text{dst})$into a one-dimensional integer. The purpose of this transformation is to assign a unique number to each stream, "thus transforming the problem of hashing pairs of data into a problem of hashing integers. We adopt the classic Cantor pairing function ($\pi$) for this task. For any stream $f$, its unique stream identifier $id(f)$is defined as follows:
\begin{equation}
    id(f) = \pi(\text{src}, \text{dst}) = (\text{src} + \text{dst})(\text{src} + \text{dst} + 1)/2 + \text{dst}
\end{equation}

To ensure the validity of this step, we have to show that the function is injective, that is, any distinct flow pair $(s,d)$ necessarily yields a unique identifier.
\begin{theorem}
For any two pairs of natural number $(s_1, d_1) $and $(s_2, d_2) $, if $\pi (s_1, d_1) = \pi (s_2, d_2) $, there must be $s_1= s_2 $ and $d_1 = d_2 $.
\end{theorem}

\begin{proof}
The proof begins by assuming that the two pairs of input function values are equal:
$\pi(s_1, d_1) = \pi(s_2, d_2)$.
Auxiliary variables are introduced to simplify the analysis. Let $w = s+d$ and define the $w_{th}$ triangular number $T(w)$ as follows.
\begin{equation}
    T(w) = \frac{1}{2}w(w+1)
\end{equation}
The pairing function can then be written as 
\begin{equation}
    \pi(s,d) = T(s+d) +d 
\end{equation}
The premise then becomes:
\begin{equation}
    T(s_1+d_1) + d_1 = T(s_2+d_2) + d_2
\end{equation}
Let $w_1 = s_1+d_1$ and $w_2 = s_2+d_2$, then the premise is:
\begin{equation}
    T(w_1) + d_1 = T(w_2) + d_2
\end{equation}

A key inequality on the value of $\pi(s,d)$ needs to be established. By definition, $s \ge 0$ and $d \ge 0$, so we have $d \le s+d = w$. The bound for $\pi(s,d)$ follows.
\begin{equation}
    T(w) \le \pi(s,d) \le T(w) + w
\end{equation}

Next, let's look at the relationship between two consecutive triangular numbers:
\begin{equation}
    T(w+1) = \frac{1}{2}(w+1)(w+2) = T(w) + w + 1
\end{equation}

It can be deduced that:
\begin{equation}
    \pi(s,d) \le T(w) + w < T(w)+w+1 = T(w+1)
\end{equation}

This leads to the central inequality in the proof:
\begin{equation}
    T(w) \le \pi(s,d) < T(w+1)
\end{equation}
the inequality states that all pairs of numbers satisfying $s+d=w$ have pairing function values that fall strictly between two consecutive triangular numbers $T(w)$ and $T(w+1)$

Then, dicuss the relationship between $w_1$ and $w_2$:

\textit{Case 1: $w_1 = w_2$}

If $w_1 = w_2$, then $T(w_1) = T(w_2)$. In the premise equation $T(w_1) + d_1 = T(w_2) + d_2$, by eliminating the terms $T(w)$ from both sides, we directly obtain the following:
\begin{equation}
    d_1 = d_2
\end{equation}

And because $s_1 + d_1 = w_1 = w_2 = s_2 + d_2 $, under the condition of $d_1 = d_2 $, there must be:
\begin{equation}
    s_1=s_2
\end{equation}

In this case, we prove that $(s_1, d_1) = (s_2, d_2)$.

\textit{Case 2: $w_1 \neq w_2$}

This case is proved by contradiction. Without loss of generality, assume that $w_1 > w_2$. Since $w_1$, $w_2$ are both integers, this is equivalent to $w_1 \ge w_2+1$. According to the derived core inequality:
\begin{equation}
    \pi(s_1, d_1) \ge T(w_1)
\end{equation}

Meanwhile, for $(s_2, d_2)$, it follows:
\begin{equation}
    \pi(s_2, d_2) < T(w_2+1)
\end{equation}
since the trigonometric function $T(w)$is a strictly monotone increasing function with respect to $w$and by assumption $w_1 \ge w_2+1$, it follows:
\begin{equation}
    T(w_1) \ge T(w_2+1)
\end{equation}

As a result, it follows:
\begin{equation}
\pi(s_1, d_1) \ge T(w_1) \ge T(w_2+1) > \pi(s_2, d_2)    
\end{equation}

In this case, we conclude that $\pi(s_1, d_1) > \pi(s_2, d_2)$. This contradicts the initial premise $\pi(s_1, d_1) = \pi(s_2, d_2)$. Thus, the assumption $w_1 > w_2$ does not hold. The same can be said to prove that $w_2 > w_1$ also does not hold.

In summary, the case $w_1 \neq w_2$ cannot occur. The only possible is $w_1 = w_2 $, which is bound to lead to a $(s_1, d_1) = (s_2, d_2) $. Thus, the Cantor pairing function is injective, providing a unique integer identifier for each flow pair $(s,d)$.
\end{proof}
\subsubsection{Minimal perfect hash seed}




The problem, the resulting hash address space must be as small as possible, can be transformed into: Given a known set $\mathcal{I}_{k,t}$ of distinct integers, how to find an minimal perfect hash "seed"—that is, a natural number $h_{k,t}$—that makes the function $H_{k,t}(f) = id(f) \pmod{h_{k,t}}$ a minimal perfect hash function.

The collision-free condition requires that for any $id_a, id_b \in \mathcal{I}_{k,t}$ where $id_a \neq id_b$, the following must hold:
\begin{equation}
id_a \pmod{h_{k,t}} \neq id_b \pmod{h_{k,t}}
\end{equation}
Mathematically, this is equivalent to the condition that their difference is not divisible by $h_{k,t}$:
\begin{equation}
id_a - id_b \not\equiv 0 \pmod{h_{k,t}}
\end{equation}
This implies that $h_{k,t}$ cannot be a divisor of the absolute difference $|id_a - id_b|$ between any two distinct identifiers in the set $\mathcal{I}_{k,t}$. To achieve resource optimality, it is necessary to find the smallest natural number $h_{k,t}$ that satisfies this condition.

According to the Pigeonhole Principle \cite{pp1}, in order to map $n_{k,t}$ distinct numbers to $n_{k,t}$ distinct hash addresses, the size of the address space (i.e., the modulus $h_{k,t}$) must be at least equal to the number of elements. Therefore, an absolute lower bound for the minimal perfect modulus can be established:
\begin{equation}
h_{k,t} \ge n_{k,t}
\end{equation}

Based on this, the ground controller can perform a deterministic search algorithm to find the minimal perfect modulus $h_{k,t}$: 
Algorithm from the theory of lower bound $h = n_ {k, I} $, successive increasing upward ($h = n_ {k, I}, n_ {k} I + 1, n_ {k, I} + 2, \ dots $). At each step, check whether the current $h$ satisfies the conflict-free condition (that is, check whether there exists $id_a, id_b \in \mathcal{I}_{k,t}$ such that $h$ is divisible by $| id_A-id_b |$). The first $h$that satisfies the condition found by the algorithm is the minimum perfect hash modulus $h_{k,t}$ sought.



\subsection{Deployment on Satellites}
To achieve efficient and resource-saving measurement in LEO mega-constellation networks, CountingStars (CS) designs a lightweight on-board architecture optimized for the high dynamic topologies and limited on-board resources in PBLB scenarios. We introduce the port aggregation data structure, which assigns an aggregation counter to each flow and logically segments it into multiple subfields, with each subfield corresponding to a specific output port. This section elaborates on the on-board architecture design from two perspectives: flow processing pipeline and resource utilization analysis.

\begin{figure}[!htbp]
	\centering
	\includegraphics[width=\linewidth]{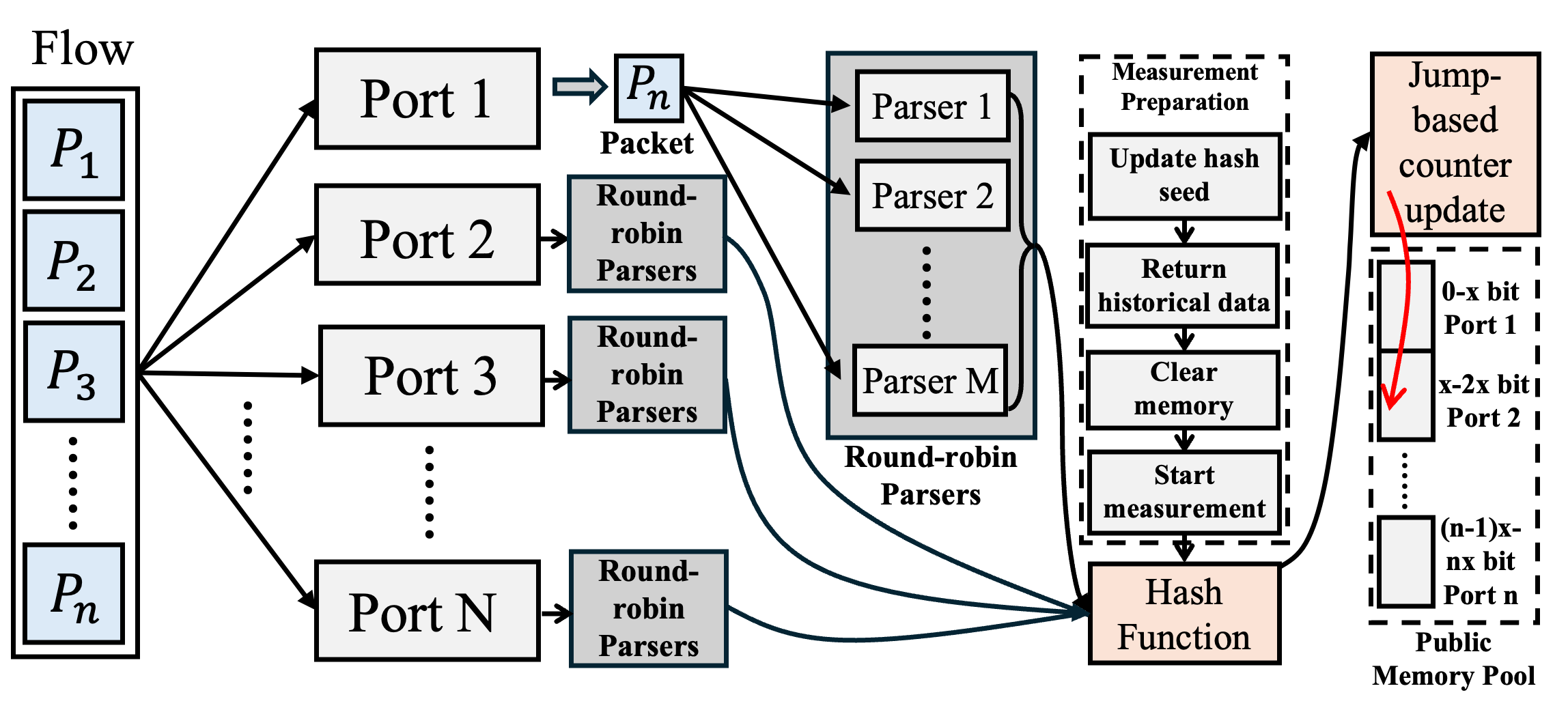}
	\caption{Deployment on Satellite}
	\label{Deployment}
\end{figure}
\subsubsection{Flow Processing Pipeline}

\textbf{Packet Parsing and Flow Identifier Generation.} In LEO satellite networks, data flows are distributed across four output ports of a satellite node via PBLB, with each port equipped with $M$ parsers to handle high-throughput traffic in parallel. Upon entering a port, packets are sequentially assigned to a parser by round-robin mechanism. Each parser extracts the source address ($ \text{src} $) and destination address ($ \text{dst} $) from the packet header, forming the flow identifier pair ($ \text{src}, \text{dst} $). Subsequently, the parser employs the Cantor pairing function to map the two-dimensional flow identifier pair to a unique natural number identifier $ t $, defined as:
\begin{equation}
    t = \pi(\text{src}, \text{dst}) = \frac{1}{2}(\text{src} + \text{dst})(\text{src} + \text{dst} + 1) + \text{dst}
\end{equation}
this function ensures that each flow identifier pair maps to a unique $ t $, simplifying flow identification to the processing of a single integer, thus facilitating subsequent dynamic hash and address mapping.


\textbf{Dynamic Hash and Address Mapping.} After obtaining the unique flow identifier $ t $, the system performs a hashing operation using a minimal perfect hash seed ($ h_{k,i} $) pre-distributed by the ground controller. Specifically, the on-board architecture computes a hash key $ m $ by applying a modulo operation on $ t $ with the hash seed $ h_{k,i} $ for the current period:
\begin{equation}
    m = t \mod h_{k,i}
\end{equation}
the hash key $ m $, combined with a fixed basic address, determines the storage address of the corresponding counter. This dynamic hash mechanism ensures collision-free address mapping within a specific period, mitigating hash collisions caused by topological changes.

\textbf{Counter Update and Memory Management.} Upon determining the counter address, the system performs a jump-based counter update based on the packet’s output port number. Each counter is a 64-bit register, partitioned into four 16-bit subfields, each corresponding to the traffic statistics of one of the four output ports. The update mechanism is as follows: for a packet transmitted through port $ p $ ($ p = 1, 2, 3, 4 $), the counter increments by $ 2^{16 \times (p-1)} $. Through jump-based counter update, this design efficiently encodes port-level traffic information into a single counter, avoiding the need for separate data structures per port and thus mitigating memory inflation.

\textbf{Data Transmission and Robustness Assurance.} At the end of each measurement period, the satellite node transmits the contents of all counters to the ground controller via the inter-satellite link. To enhance system robustness, the on-board architecture retains the measurement data from the previous period until a new hash seed is received. Upon receiving the new seed, the data is retransmitted, and the counters are cleared. This mechanism ensures data integrity, significantly reducing the risk of data loss due to link interruptions or seed update delays, thereby maintaining reliable measurement performance in highly dynamic LEO networks.

\subsubsection{Resource Utilization Analysis}


The computational overhead of the on-board architecture primarily stems from ALU operations across the flow processing stages \cite{alu1, alu2}. In the packet parsing stage, the parser performs two logical AND operations to extract $ \text{src} $ and $ \text{dst} $, accompanied by two SRAM accesses to read packet data. For computing $ t $ using the Cantor pairing function, the ALU executes one addition ($ \text{src} + \text{dst} $), one multiplication ($ \frac{1}{2}(\text{src} + \text{dst})(\text{src} + \text{dst} + 1) $), and one shift operation (right shift by one bit), without additional memory access. In the hashing stage, the ALU performs one modulo operation, accompanied by one SRAM access to retrieve the hash seed and determine the counter address. In the counter update stage, the ALU executes one addition ($ 2^{16 \times (p-1)} $) and one SRAM accesses (write to the counter). Overall, the ALU operations are efficient, with minimal memory accesses, ensuring low computational overhead.

\section{Experiments}
To evaluate the proposed CountingStars in a real-world environment, we developed a customized software simulation platform. This simulation platform reproduces the typical characteristics of LEO mega-constellation networks, particularly their highly dynamic topology and time-varying traffic demands. It also provides greater flexibility in terms of hash quantity and data structure size, enabling a better understanding of the operation of our system. The source code of the platform can all be obtained on Github\cite{Platform}. In addition, we deploy CountingStars on FPGA and test its performance. 

\subsection{Simulation Platform} 

Our simulation platform consists of two core components: a digital twins system (detailed in \ref{Digital Twins}) and a discrete-event-driven network simulator. At the beginning of each discrete time step, our network simulator loads a time-stamped adjacency matrix and a corresponding satellite-level traffic matrix from the digital twins system, accurately reproducing the continuous topological changes and ensuing traffic pattern shifts. To simulate network behavior with high fidelity, each satellite is modeled as a node with packet processing capabilities and configurable resources, while ISLs are modeled as channels with specific bandwidth, delay, and buffers. The simulator dynamically responds to topological changes, implementing a lossless handling mechanism that preserves and instantly reroutes any in-flight packets affected by disruptions, while concurrently updating the forwarding paths of all active flows.

We deployed CountingStars and several mainstream comparison schemes using this simulation platform, and monitored the traffic of two typical low-orbit constellation networks, Starlink and Iridium. Finally, we evaluated the effectiveness of each scheme based on the collected performance indicators.

\subsection{Evaluation Metrics} Three widely used metrics are applied to evaluate the performance of CountingStars' measurement. The calculation equations for these metrics are presented as follows:

\begin{itemize}
    \item \textbf{Average Relative Error (ARE):}
    
    \begin{equation}
        ARE = \frac{1}{n} \sum_{i=1}^{n} \frac{|f_i - \hat{f}_i|}{f_i}
    \end{equation}
    
    where $n$ is the number of flows, and $f_i$ and $\hat{f}_i$ are the actual and estimated flow sizes respectively.

    \item \textbf{Weighted Mean Relative Error (WMRE):}
    
    \begin{equation}
        WMRE = \frac{\sum_{i=1}^{z} |n_i - \hat{n}_i|}{\sum_{i=1}^{z} (\frac{n_i + \hat{n}_i}{2})}
    \end{equation}
    
    where $z$ is the maximum flow size, and $n_i$ and $\hat{n}_i$ are the true and estimated numbers of flows of size i respectively.

    \item \textbf{Relative Error (RE):}
    
    \begin{equation}
        RE = \frac{|\textit{True} - \textit{Estimated}|}{\textit{True}}
    \end{equation}
    
    where \textit{True} and \textit{Estimated} are the true and estimated values, respectively.
\end{itemize}

For both three metrics, the smaller they get to 0, the higher the accuracy of the measurement.
\subsection{Baselines} To evaluate CountingStars, we select three advanced network measurement schemes — \textbf{CountingStars (CS)}, Elastic Sketch (ES), and FlowLIDAR — representing state-of-the-art techniques for terrestrial network measurement.

\begin{itemize}
\item \textbf{Count-Min Sketch}: Employs a 2D array with pairwise independent hash functions, supporting point, range, and inner product queries with $O(1/\varepsilon)$space complexity for traffic statistics.
\item \textbf{Elastic Sketch}: Separates heavy and light flows, using a hash table for key flows and Count-Min Sketch for light flows, leveraging diverse methods to adapt to varying flow distributions.
\item \textbf{FlowLIDAR}: Integrates Bloom filters with Count-Min Sketch to detect new flows in the data plane, forwarding them to the control plane for measurement.
\end{itemize}
\vspace{-0.1cm}
\subsection{Scenarios}
\vspace{-0.1cm}
\subsubsection{Constellation Setup}
The experiments target two LEO mega-constellation networks: Iridium \cite{iri} and Starlink\cite{sta}. Iridium comprises 66 satellites evenly distributed across 6 orbital planes, while Starlink first-generation includes 1584 satellites across 72 orbital planes. To capture high-dynamic topology changes, we select 100 time slices (1-second intervals), totaling 100 seconds, to simulate network evolution.

\subsubsection{Build of Traffic Data}
To evaluate CountingStars (CS) across different constellation scales, we generate a simulation dataset using the digital twins system.
This process begins by modeling the geographical distribution of the network of ground stations and user terminals, then simulating the sequence of the traffic matrix $F_{\text{ter}}^{t}$ between them. The characteristics of the traffic load carried by satellite networks are intricate. Spatially, the distribution of ground stations is highly uneven due to factors such as topography considerations. 

Access traffic in extreme environments like oceans, deserts, or polar regions is significantly lower than in densely populated, profitable areas, leading directly to a non-uniform spatial distribution of traffic within the satellite network \cite{spatial1, spatial2, spatial3, spatial4}. Temporally, the global distribution of ground stations across different time zones results in a tidal effect as traffic fluctuates with the day-night cycle. \cite{temp1, temp2, temp3, temp4}

Therefore, unlike many simplified models, real-world Internet traffic does not follow a simple Poisson distribution \cite{poi1, poi2, poi3, poi4}. To replicate these key characteristics in our simulation, we adopt a more sophisticated model to generate traffic demands. For any given time $t$, the total traffic volume $D_t$ sent by all ground controllers is calculated as:

\begin{equation}
    D_{t} = offerload \times B \times n_{ter}
\end{equation}
where $offerload$ is the network load factor, $B$ is the maximum ISL bandwidth, and $n_{ter}$ is the total number of ground controller. The local time $t_m$ for a ground controller $i$ with longitude $x_i$ at GMT time $t$ is determined by:

\begin{equation}
    t_{m} = \lfloor\frac{t}{3600}\rfloor + \lfloor\frac{x_{i}}{15}\rfloor
\end{equation}

Then, based on a normalized 24-hour cumulative load profile, a corresponding traffic intensity weight $w_m^t$ is obtained, where $w_m^t = w_{t_m} / w_{\text{total}}$. The traffic sent from a single ground controller $i$ within that time zone is given by:

\begin{equation}
    (f_{m}^{t})_{i} = \frac{D_{t} \times w_{m}^{t}}{n_{m}}
\end{equation}
where $n_m$ is the number of ground controllers in that period, and the sum of ground controllers across all periods satisfies:
\vspace{-0.1cm}
\begin{equation}
    \sum_{m=0}^{23} n_m = n_{ter}
\end{equation}

Assuming destination nodes are selected uniformly, the traffic from a source ground controller $i$ to a destination ground controller $j$, $F(i,j)^t$, is modeled as:

\begin{equation}
    F(i,j)^{t} = U(0.1,1) \times (f_{m}^{t})_{i}
\end{equation}
where $U(0.1,1)$ represents a uniform random distribution. By iterating through all pairs of ground controller, the complete terrestrial traffic matrix is formed and ultimately mapped to a satellite-level traffic matrix sequence $F_{\text{sat}}^{t}$.

\begin{figure}[!htbp]
	\centering
	\includegraphics[width=\linewidth]{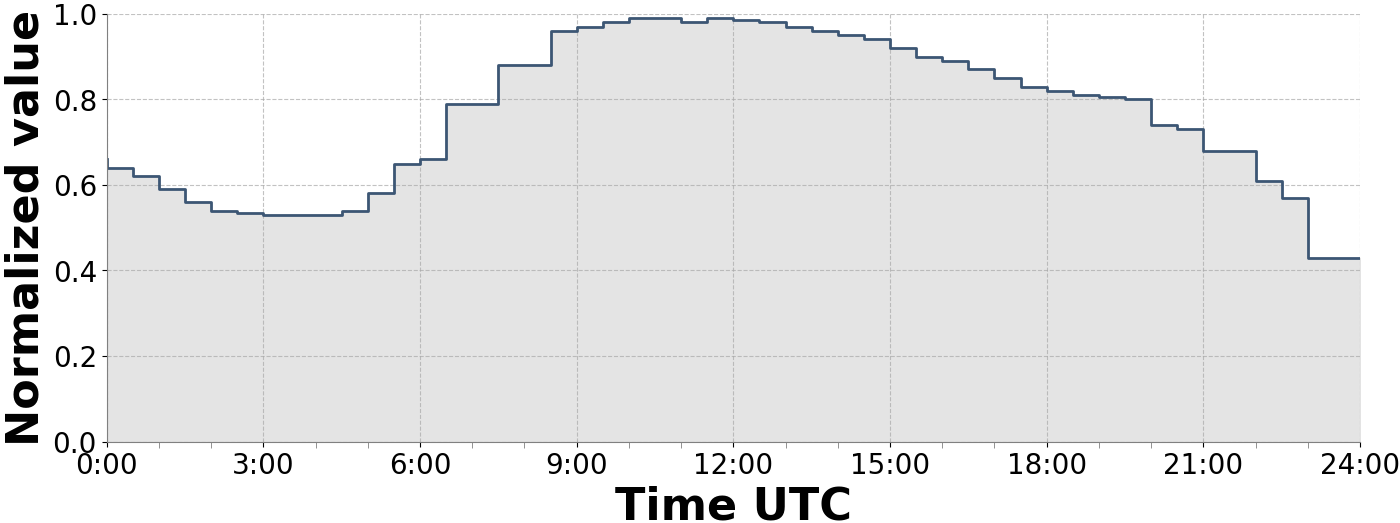}
	\caption{Ground controller modeling}
	\label{Ground controller modeling}
\end{figure}

By routing each ground-to-ground flows from its access satellite to its exit satellite, we map the ground traffic matrix to the satellite topology, thereby obtaining the time series of the inter-satellite traffic matrices dataset.

The obtained dataset include total transmitted packets, total flow count, duration, loads (0.1, 0.5, 0.9) and inter-satellite traffic matrix. Data is collected at 1-second intervals, reflecting dynamic traffic patterns in PBLB scenarios. Key performance metrics for Iridium and Starlink constellations under varying loads are summarized in Table~\ref{tab:network_performance}.

\subsubsection{Memory Setup}
To ensure a fair comparison, we standardize the memory configurations in our experiments. For all baseline schemes, we follow the recommended settings from their respective papers, setting the counter size to 32-bit. For our proposed CountingStars scheme, we leverage its unique port-aggregated data structure, where the counter for each flow is a 64-bit register, partitioned into multiple subfields to simultaneously record the traffic of a single flow across multiple output ports. In defining flow size, we assume the minimal packet is 64 bytes; an incoming packet of 120 bytes is considered as $\lceil\frac{120}{64}\rceil=2$ packets.

\begin{table}[ht]
\centering
\small 
\setlength{\tabcolsep}{3pt} 
\caption{Dataset for Iridium and Starlink Constellations}
\label{tab:network_performance}
\resizebox{\columnwidth}{!}{%
\begin{tabular}{c c c c c}
\toprule
Constellation & Load & Duration & \# of Packets & \# of Flows \\
\midrule
Iridium & 0.1 & \SI{100}{\second} & 3.7K & 0.4K \\
Iridium & 0.5 & \SI{100}{\second} & 22.4K & 1.9K \\
Iridium & 0.9 & \SI{100}{\second} & 42.5K & 3.7K \\
Starlink & 0.1 & \SI{100}{\second} & 12.1M & 317.7K \\
Starlink & 0.5 & \SI{100}{\second} & 66.8M & 932.1K \\
Starlink & 0.9 & \SI{100}{\second} & 117.5M & 1.3M \\
\bottomrule
\end{tabular}}
\end{table}

\subsection{Memory Usage and Performance}

\begin{figure*}[ht!]
    \centering

    \includegraphics[width=0.6\textwidth]{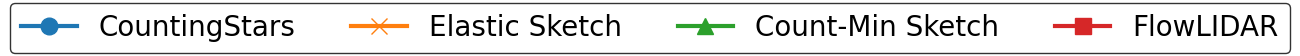}

    \vspace{0.1cm} 

    \begin{subfigure}[b]{0.49\textwidth}
        \centering
        \includegraphics[width=\textwidth]{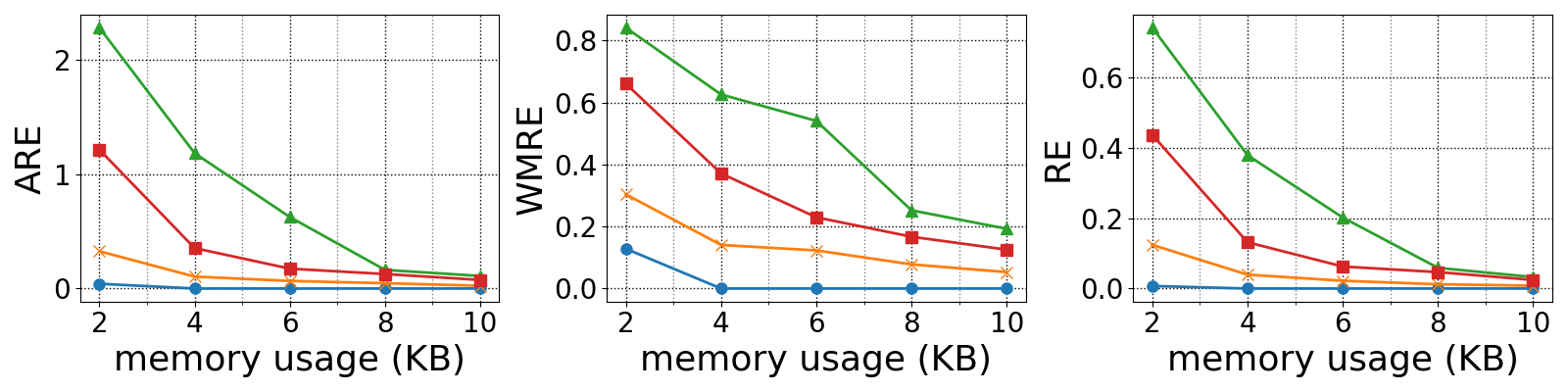}
        \caption{Offerload 0.1 under Iridium scenario}
        \label{fig:iridium_0.1}
    \end{subfigure}
    \hfill
    \begin{subfigure}[b]{0.49\textwidth}
        \centering
        \includegraphics[width=\textwidth]{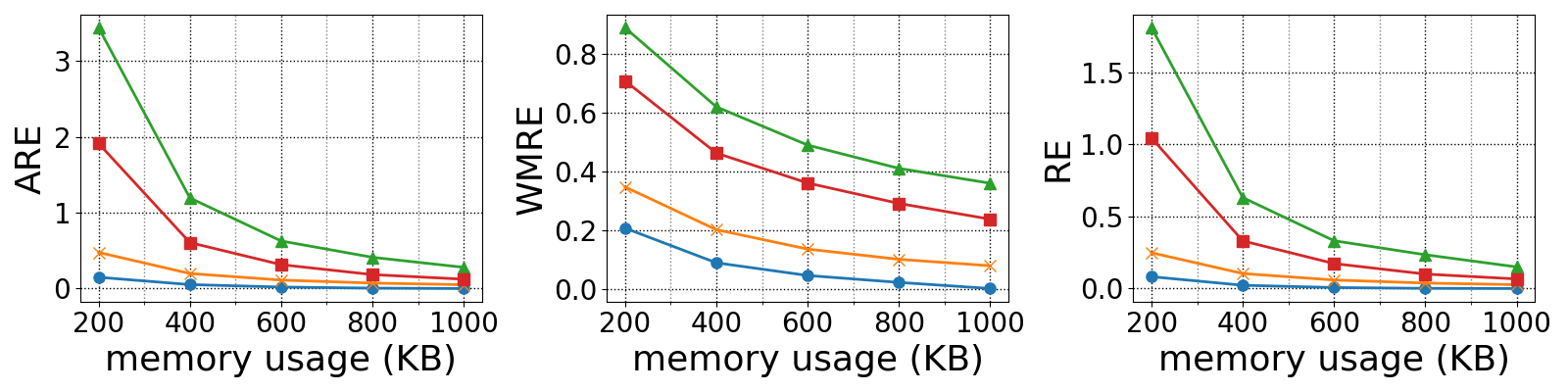}
        \caption{Offerload 0.1 under Starlink scenario}
        \label{fig:starlink_0.1}
    \end{subfigure}

    \vspace{0.3cm} 

    \begin{subfigure}[b]{0.49\textwidth}
        \centering
        \includegraphics[width=\textwidth]{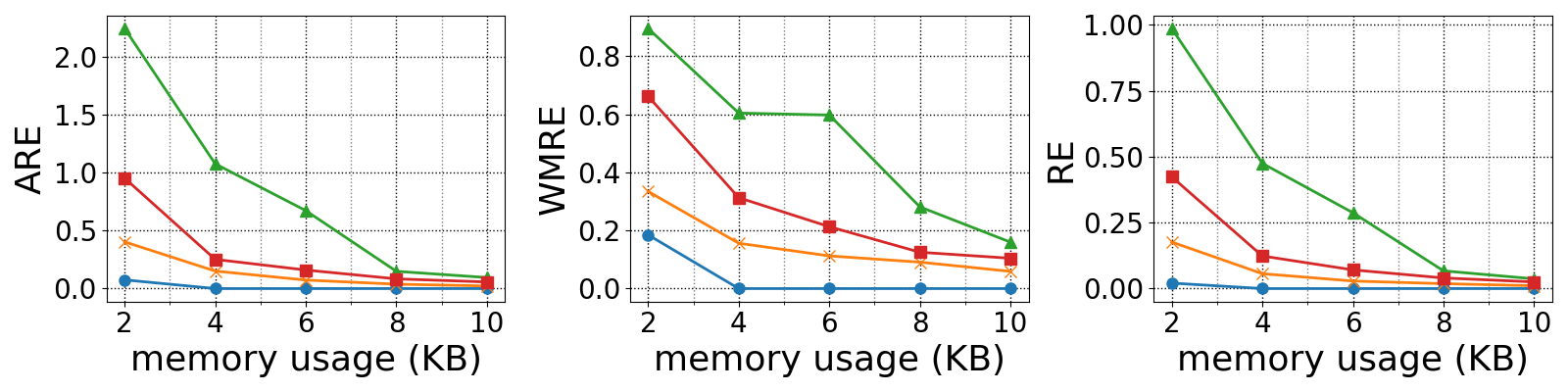}
        \caption{Offerload 0.5 under Iridium scenario}
        \label{fig:iridium_0.5}
    \end{subfigure}
    \hfill
    \begin{subfigure}[b]{0.49\textwidth}
        \centering
        \includegraphics[width=\textwidth]{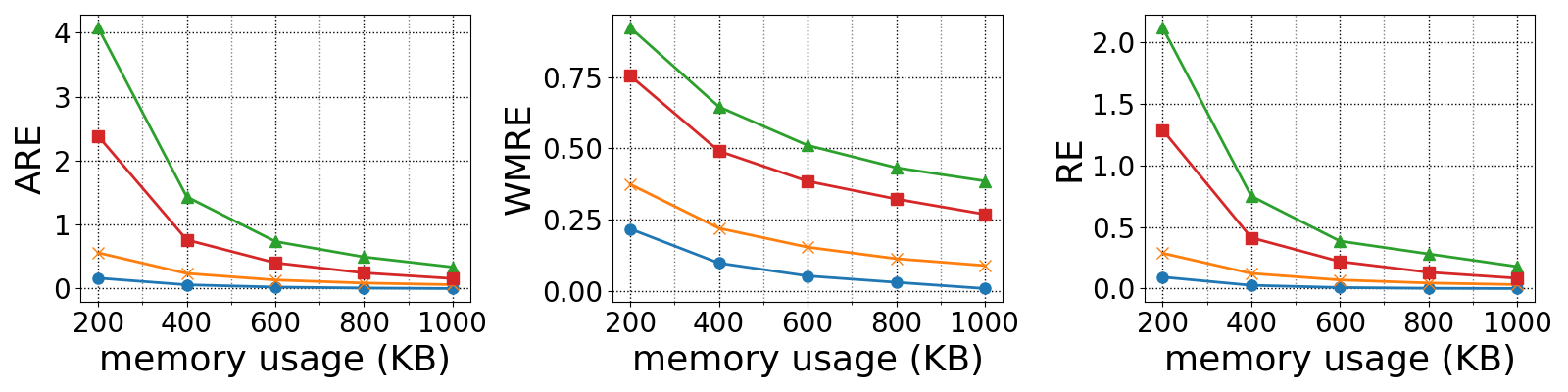}
        \caption{Offerload 0.5 under Starlink scenario}
        \label{fig:starlink_0.5}
    \end{subfigure}

    \vspace{0.3cm} 

    \begin{subfigure}[b]{0.49\textwidth}
        \centering
        \includegraphics[width=\textwidth]{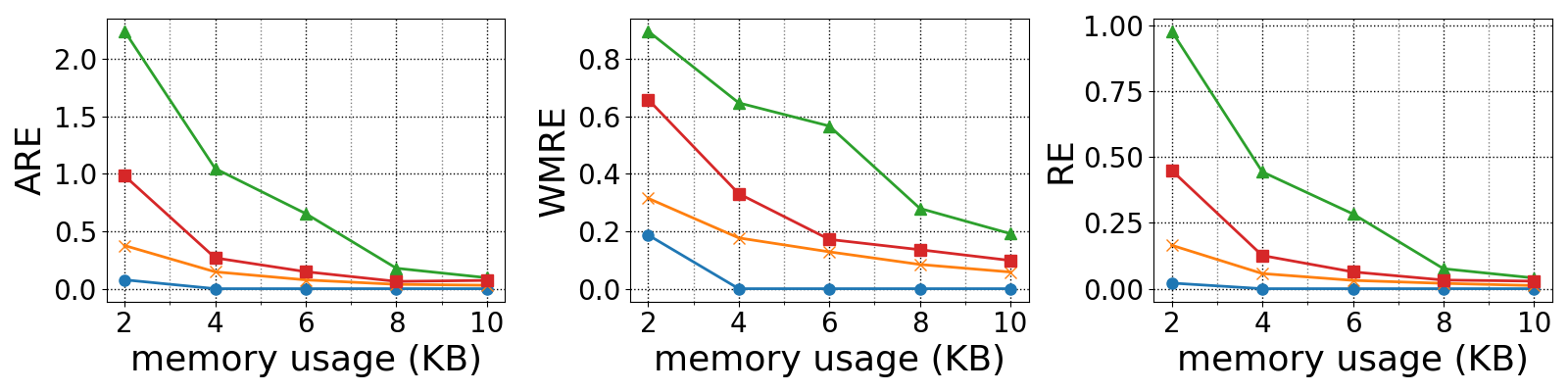}
        \caption{Offerload 0.9 under Iridium scenario}
        \label{fig:iridium_1.0}
    \end{subfigure}
    \hfill
    \begin{subfigure}[b]{0.49\textwidth}
        \centering
        \includegraphics[width=\textwidth]{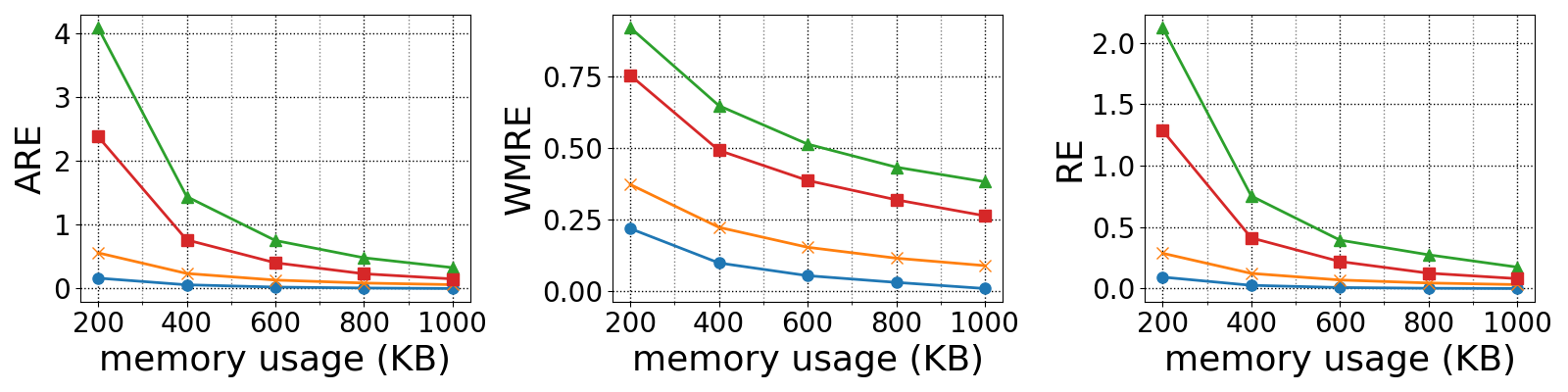}
        \caption{Offerload 0.9 under Starlink scenario}
        \label{fig:starlink_1.0}
    \end{subfigure}

    \caption{Comparison of Iridium and Starlink loads at different levels}
    \label{fig:combined}
\end{figure*}

To evaluate the measurement performance of CountingStars in LEO mega-constellation networks, we developed a simulation platform to compare CS against three terrestrial network measurement schemes: CM, ES, and FlowLIDAR. The experiments encompass the Iridium (66 satellites) and Starlink (1584 satellites) constellations, with traffic loads set at 0.1, 0.5, and 0.9, forming six experimental groups. 

Tests were conducted under varying memory allocations (Iridium: 2–10 KB; Starlink: 200–1000 KB), using Average Relative Error (ARE), Weighted Mean Relative Error (WMRE), and Relative Error (RE) to assess fine-grained measurement accuracy. Results are presented in Figure. \ref{fig:combined}, comprising subfigures \ref{fig:iridium_0.1} – \ref{fig:starlink_1.0}, corresponding to Iridium (0.1, 0.5, 0.9 loads) and Starlink (0.1, 0.5, 0.9 loads), respectively.


The experimental results demonstrate that CS significantly outperforms the comparison schemes across all scenarios, exhibiting superior measurement accuracy and memory efficiency. Notably, under the Iridium constellation with 0.1 load (2 KB memory, Figure.~\ref{fig:iridium_0.1}), CS achieves an ARE of 0.0418, WMRE of 0.1282, and RE of 0.0071, markedly lower than CM (2.2841, 0.8413, 0.7419), ES (0.3266, 0.3037, 0.1241), and FlowLIDAR (1.2156, 0.6605, 0.4368), with ARE reductions of 98.2\%, 87.2\%, and 96.6\%, respectively. At higher loads (Iridium, 0.9 load, 8 KB, Figure.~\ref{fig:iridium_1.0}), the ARE, WMRE, and RE of CS near 0, surpassing CM (0.1781, 0.2796, 0.0758), ES (0.0378, 0.0844, 0.0202), and FlowLIDAR (0.0636, 0.1358, 0.0328). 


The Starlink experiments further validate CS’s robustness; for instance, at 0.5 load (400 KB, Figure.~\ref{fig:starlink_0.5}), CS records an ARE of 0.0589, WMRE of 0.0967, and RE of 0.0261, lower than CM (1.4359, 0.6459, 0.7506), ES (0.2367, 0.2192, 0.1231), and FlowLIDAR (0.7624, 0.4903, 0.4132). At the maximum load and memory (Starlink, 0.9 load, 1000 KB, Figure.~\ref{fig:starlink_1.0}), CS achieves an ARE of 0.0007, WMRE of 0.0073, and RE of 0.0002, far below CM (0.3284, 0.3821, 0.1747, ARE reduced by 99.8\%), ES (0.0612, 0.0877, 0.0321, ARE reduced by 98.9\%), and FlowLIDAR (0.1512, 0.2625, 0.0809, ARE reduced by 99.5\%). These results underscore CS’s exceptional performance in both small-scale (Iridium) and large-scale (Starlink) satellite network scenarios.

CS demonstrates a remarkable balance between memory efficiency and measurement accuracy. In high-accuracy scenarios, to achieve an ARE of approximately 0.05 under 0.9 load, CS requires only 4 KB in Iridium, while CM demands 16 KB (saving 75\%), and both ES and FlowLIDAR require 8 KB (saving 50\%). In Starlink, to attain an ARE of approximately 0.15 under 0.9 load, CS needs only 200 KB, whereas CM requires far beyond 1000 KB (saving over 80\%), ES needs 600 KB (saving 67\%), and FlowLIDAR needs 1000 KB (saving 80\%).  In resource-constrained scenarios, such as Iridium at 2 KB and Starlink at 200 KB, CS’s WMRE and RE are consistently 2–10 times lower than those of the baseline schemes, highlighting its efficiency in low-memory environments. 


These results reveal the inherent limitations of sketch-based methods relying on static hash arrays. CM and FlowLIDAR implement traffic statistics based on flow occurrence probability estimation through multi-row hash arrays but rely on multiple data structure instantiations, leading to substantial memory inflation. In contrast, CS requires only a single shared entry and a compact counter array, without necessitating multiple data structure instantiations, and further employs counter shift operations to refine traffic information to the port level, thereby reducing memory requirements in port-level measurement.

Although ES still cannot escape the constraints of multiple instantiations, it uses an innovative heavy-light flow data structure to address the limitation of CM and FlowLIDAR, which require large-scale memory to ensure accuracy. Unfortunately, in LEO mega-constellation networks, flows often cannot be clearly distinguished as heavy or light, leading to frequent misclassifications of flow sizes and consequent loss of critical flow-specific information.

Additionally, the static hash designs of these three methods fail to adapt to the dynamic flow distributions in LEO networks, where the same counter address is prone to being occupied by the same flow ID at different time points, causing hash collisions and significantly reducing measurement accuracy. CS, through topology-aware dynamic hash combined with digital twins-predicted flow sets, assigns unique hash addresses to each flow within specific time periods, effectively mitigating hash collisions.

\subsection{FPGA Implementation}
To validate the performance of CountingStars in LEO mega-constellation networks, we implemented the CountingStars scheme on an FPGA, using an Xilinx Artix-7 series XC7A100T-2FGG484I board with a clock frequency of 50 MHz. In the experiments, we reproduced three baseline methods---CM, ES---ensuring a fair comparison.

We tested the single-packet processing cycle $ T $ and throughput for all methods under the condition of processing one packet per clock cycle (1-packet-per-cycle, 1-PPC). The single-packet processing cycle $ T $ is a critical metric for evaluating the efficiency of on-board measurement systems, reflecting processing efficiency under identical throughput conditions. In the resource-constrained LEO environments, a smaller $ T $ enables higher processing rates with lower memory and computational overhead, thereby enhancing system throughput. The test results are shown in Table~\ref{tab:T_comparison} that CS has achieved $T = 136 \, \text{ns} $, close to CM ($101 \, \text{ns} $) or ES ($143 \, \text{ns} $). It reveals the possibility of CS to deploy on real satellites.

\begin{table}[ht]
\centering
\small 
\setlength{\tabcolsep}{4pt} 
\caption{Single-Packet Processing Cycle ($T$) Comparison}
\label{tab:T_comparison}
\begin{tabular}{c | c | c}
\toprule
Method & Processing Time $T$& Throughput\\ 
\midrule
CountingStars (CS) & 136ns & 48.46Mpps\\
Count-Min Sketch (CM) & 101ns& 45.23Mpps\\
Elastic Sketch (ES) & 143ns & 50.00Mpps\\
\bottomrule
\end{tabular}
\end{table}


In terms of throughput, the measurement method with mature pipeline mechanism will perform better in real hardware deployment. The throughput of CM, CS and ES are close to 50Mpps, which is the theoretical upper limit of 50 MHz clock frequency for this board.

\section{Conclusion}
We propose the CountingStars framework, designed to address memory inflation caused by PBLB and hash collisions induced by highly dynamic topologies in LEO networks. The framework leverages a digital twins to predict topologies and introduces a dynamic hash mechanism to eliminate collisions. Additionally, it designs a compact port aggregation data structure, significantly reducing memory usage. Experiments demonstrate that the scheme reduces memory consumption by 50–80\% and measurement error by 75–99.8\%, achieving both high accuracy and low memory overhead.

\end{document}